  \documentclass[10pt]{article}
  \usepackage{amsmath,amsthm,amsfonts,latexsym,epigraph,accents}

\usepackage{amssymb}
\usepackage{enumitem}
\usepackage{verbatim}
\usepackage[usenames]{color}
\usepackage{url}
\usepackage[pdfpagemode=UseNone,pdfstartview=FitH]{hyperref}

  \newtheorem{theorem}{Theorem}[section]
  
  \newtheorem{proposition}[theorem]{Proposition}
  \newtheorem{lemma}[theorem]{Lemma}
  \newtheorem{corollary}[theorem]{Corollary}

  \theoremstyle{definition}

  \newtheorem{definition}[theorem]{Definition}
  \newtheorem{example}[theorem]{Example}

  \theoremstyle{remark}
  \newtheorem{remark}[theorem]{Remark}

\newcommand\xqed[1]{%
  \leavevmode\unskip\penalty9999 \hbox{}\nobreak\hfill
  \quad\hbox{#1}}
\newcommand\aqed{\xqed{$\triangleleft$}}

\newcommand{\1}{^{-1}}

\newcommand{\B}{\mathcal{B}}

\newcommand{\eps}{\varepsilon}

  \newcommand{\f}{\hat{f}}

\newcommand{\R}{\mathbb{R}}
\DeclareMathOperator{\Range}{Range}
\newcommand{\Q}{\mathbb{Q}}
\newcommand{\T}{\mathcal T}

  \newcommand{\U}{\mathcal{U}}

\renewcommand{\P}{\mathbf{P}}
\renewcommand{\phi}{\varphi}

\newlength{\dhatheight}
\newcommand{\hathat}[1]{%
  \settoheight{\dhatheight}{\ensuremath{\hat{#1}}}%
  \addtolength{\dhatheight}{-0.35ex}%
  \hat{\vphantom{\rule{1pt}{\dhatheight}}%
  \smash{\hat{#1}}}}

  \title{Test statistics and p-values}
  \author{Yuri Gurevich and Vladimir Vovk}

\begin{document}

\maketitle

\begin{abstract}
We point out that the traditional notion of test statistic is too narrow, and
we propose a natural generalization that is arguably maximal. 
The study is restricted to simple statistical hypotheses.

    \bigskip

    \noindent
    The version of this paper at
    \url{http://www.alrw.net/} (Working Paper 16) is updated more often.
\end{abstract}

  \setlength{\epigraphwidth}{0.65\textwidth}

  \epigraph{\ldots  the enormous usefulness of mathematics in the natural sciences is something bordering on the mysterious\ldots}
  {Eugene P. Wigner, 1960}

  \epigraph{\ldots statistics is a branch of applied mathematics, like symbolic logic or hydrodynamics.
    Examination of foundations is desirable,
    but it must be remembered that undue emphasis on niceties is a disease
    to which persons with mathematical training are specially prone.}
  {George A. Barnard, 1947}

\section{Introduction}

The standard definition of the p-value associated with a given test statistic $f$ and outcome $x$ is
\begin{equation}\label{eq:p}
  \f(x) = \P[f\le f(x)]
\end{equation}
where $[f\le f(x)] = \{y: f(y)\le f(x)\}$.
The standard textbook convention is that the test statistic $f$ takes values in the real line.
But the definition of $\f$ requires only that the codomain of $f$ be an ordered measurable space and that the initial segments $(-\infty,f(x)]$ be measurable.
Should we generalize the notion of test statistics?
Would any linearly ordered measurable space work as the codomain of a test statistic $f$ provided that the initial segments $(-\infty,f(x)]$ are measurable?

Our answer to the first question is ``yes''. A modest but useful generalization is already used in applied statistics, albeit implicitly. 
Our analysis turns up a more radical generalization, more natural and arguably 
the maximal generalization that makes sense.

The answer to the second question is an emphatic no.
There exist a probability space $(\Omega,\Sigma,\P)$ and a statistic $f$ with values in a linearly ordered set with measurable initial segments such that every $\f(x) = 0$.
Informally, this makes no sense: we are entitled to reject the null hypothesis whatever happens.
Formally, this contradicts the standard property
\begin{equation}\label{eq:valid}
  \P[\f \le \epsilon] \le \epsilon \qquad
  \text{ for every nonnegative }\eps<1
\end{equation}
of the validity of p-values.

\begin{example}\label{exa:omega1}
The sample space $\Omega$ is the collection (known to set-theorists as $\omega_1$) of countable (that is finite or infinite countable) ordinals.
In set-theory, every ordinal is the set of smaller ordinals: 0 is the empty set, $1 = \{0\}$, $2 = \{0,1\}$, $3 = \{0,1,2\}$,
the first infinite ordinal $\omega_0$ is the set $\{0,1,\ldots\}$ of natural numbers,
$\omega_0+1 = \omega_0 \cup \{\omega_0\}$, $\omega_0+2 = \omega_0 \cup \{\omega_0,\omega_0+1\}$, $\ldots$,
$\omega_0+ \omega_0 = \omega_0 \cup \{\omega_0 + n: n\in\omega_0\}$, and so on.
The first uncountable ordinal $\omega_1$ is the set of countable ordinals.

The $\sigma$-algebra $\Sigma$ consists of all countable subsets of $\Omega$ and their complements, and
\begin{equation*}
  \P(X)
  =
  \begin{cases}
    0 & \text{if $X$ is countable},\\
    1 & \text{if $\Omega - X$ is countable}.
  \end{cases}
\end{equation*}
Finally the statistic $f$ is the identity function: $f(x) = x$.
The order on the codomain is natural:
\[
  x < y
  \iff
  x \in y,
\]
so that
\begin{equation*}
 0 < 1 < \cdots < \omega_0 < \omega_0 + 1 < \cdots
   < \omega_0 + \omega_0 < \cdots.
\end{equation*}
For every countable ordinal $x$, the initial segment $[0,x]$ is countable. Accordingly
\[
  \hat f(x) = \P[f\le f(x)] = \P[0,x] = 0.
  \eqno\aqed
\]
\end{example}

Generalized test statistics have been, albeit implicitly, used in practical statistics. For example, to incorporate the notion of randomized p-values into definition \eqref{eq:p},
one needs generalized test statistics with codomains that are richer than the real line,
as we discuss later in the article.
We analyze what can go wrong with generalized test statistics and arrive at a generalization that is arguably the right one. 
We argue that generalized test statistics should be used more widely
and not necessarily in combination with randomization.

\section{Test statistics}

\subsection{Nominal test statistics}

\begin{definition}
  Let $\T$ be a probability space $(\Omega,\Sigma,\P)$ and $R$ any ordered measurable space with all initial segments $(-\infty,r]$ measurable.
  Any measurable function $f:\Omega\to R$ is a \emph{nominal test statistic} for $\T$.
\end{definition}

Notation $\T$ alludes to ``probability trial''. 
An ordered measurable space is a measurable space endowed with a linear order; in this article, ``ordered'' always means ``linearly ordered''.

The notion of nominal test statistic is auxiliary. As Example~\ref{exa:omega1} shows, a nominal test statistic may be unreasonable. 

To comply with the standard definition of p-values, Equation~\eqref{eq:p} above, we restrict attention to ordered measurable spaces where every initial segment of the form $(-\infty,r]$ is measurable.
Note that if $\Sigma_1, \Sigma_2$ are $\sigma$-algebras on any set $R$ and $\Sigma_1\subseteq \Sigma_2$ then every $R$-valued function that is measurable with respect to $\Sigma_2$ is measurable with respect to $\Sigma_1$. In other words, the smaller the $\sigma$-algebra, the greater the collection of measurable functions. This motivates the following definition.

\begin{definition}
The \emph{p-minimal} $\sigma$-algebra on an ordered set $R$ is the least $\sigma$-algebra on $R$ that contains every initial segment of the form $(-\infty,r]$.
\end{definition}

If $R$ is an ordered measurable space whose $\sigma$-algebra is p-minimal then the measurability requirement for $R$-valued nominal test statistics simplifies to this:
Every set $[f\le r]$ is measurable.

\begin{definition}\label{def:induced-prob}
Let $f$ be a nominal test statistic for a probability trial $(\Omega,\Sigma,\P)$ with codomain $R$. The nominal test statistic $f$ \emph{induces} the probability measure $\P_f(S) = \P(f^{-1}(S))$ on the measurable subsets of $R$.
\end{definition}

\subsection{Traditional and nearly traditional test statistics}

\begin{definition}
A nominal test statistic is a \emph{traditional test statistic} if its codomain is the real line $\R$ with the Borel $\sigma$-algebra.
\end{definition}

The real line is of course the set of real numbers with the standard order.
Its Borel $\sigma$-algebra coincides with its p-minimal $\sigma$-algebra.
In applied statistics, mostly traditional test statistics are used.

We introduce a slight generalization of traditional test statistics that may be convenient. (That is not the radical generalization mentioned in the Introduction.)
To this end, we recall a few definitions and facts. 

Every ordered set (of size at least 1) is equipped with its \emph{order topology}. The segments $(x,y)$, $(-\infty,x)$, $(y,\infty)$, and $(-\infty,\infty)$ form a base of the order topology. In this paper, the order topology is the default topology on ordered sets.

A topology is \emph{second-countable} if it has a countable base of open sets. The following proposition is proved in, e.g., \cite[Theorem~II and Lemma~3]{Cater} and (in a somewhat less explicit form) in  \cite[Theorem~24 in Section~VIII.11]{Birkhoff}. A \emph{jump} in an ordered set $(R,\le)$ is a pair $(x,y)$ of points of $(R,\le)$ such that $x<y$ and there is no $z\in R$ with $x<z<y$. 
 A mapping $f$ whose domain and codomain are both ordered sets is \emph{order-preserving} (or an \emph{embedding}) if $x<y$ implies $f(x)<f(y)$ for all $x$ and $y$ in its domain.

\begin{proposition}\label{prop:Cater}
  A linear order
  can be embedded into the real line $\R$ if and only if its order topology is second countable.
  The order topology of a linear order is second countable if and only if the topology is separable and the order has at most countably many jumps.
\end{proposition}

Notice that any embedding of a linear order with second-countable topology into the real line will be measurable:
the pre-image of an initial interval $(-\infty,r]$ of the real line
can be represented as a countable union of initial intervals of this form and is, therefore, measurable.
Now we are ready to introduce the slight generalization of traditional test statistics.

\begin{definition}\label{def:nearly}
  A \emph{nearly traditional test statistic} is a nominal test statistic whose codomain is second-countable.
\end{definition}

Traditionalists may argue that the generalization to the nearly traditional test statistics is  vacuous, and in a sense it is.
By Proposition~\ref{prop:Cater}, any nearly traditional test statistic can be composed with an order embedding to obtain a traditional test statistic.
This implies that any nearly traditional test statistic $f$ can be replaced by a traditional test statistic $f'$;
$f$ and $f'$ will be equivalent in the sense of inducing the same order on $\Omega$:
$f(x)\le f(y)$ if and only if $f'(x)\le f'(y)$.
However, as the following schematic example shows, the nearly traditional test statistic $f$ may be more convenient to work with.
This example is inspired by the literature on randomized p-values but it does not presuppose the knowledge of randomized p-values,
which will be introduced in Section~\ref{sec:randomized}.

\begin{example}\label{ex:lex}
  Let $f$ be a traditional test statistic on a discrete measurable space $(\Omega,\Sigma)$
  with $\Omega$ countable (in which case randomizing p-values becomes particularly useful).
  Let $R$ be the range of $f$ and equip $R$ with its natural order and the p-minimal (i.e., discrete in this case) $\sigma$-algebra.
  Set $\Omega'$ to the real segment $[0,1]$ and $\Sigma'$ to the p-minimal (i.e., discrete) $\sigma$-algebra on $[0,1]$.
  (In the context of randomized p-values, $\Omega'$ is interpreted as the range of random numbers
  generated by a random number generator.)
  Order $R\times[0,1]$ lexicographically, so that
  \[
    (p,r)\le (q,s) \iff p<q \text{ or }
    (p=q \text{ and } r\le s).
  \]
  It is easy to check that the lexicographic order is second-countable
  (this uses the countability of $R$)
  and that the function
  $F(p,r) = (f(p),r)$
  is a nearly traditional test statistic on the product measurable space $(\Omega\times\Omega',\Sigma\otimes\Sigma')$.
  While $F$ is rather natural, an equivalent traditional test statistic may be rather involved;
  think, e.g., of the case where $R$ is the set $\Q$ of rationals.
  \aqed
\end{example}

Radicals may argue that the generalization to the nearly traditional test statistics is too timid,
that there are natural examples of nominal test statistics with more general ordered codomains.
We agree.

\subsection{A general notion of test statistic} 

\begin{definition}
  Let $O$ and $R$ be ordered sets.
  $R$ is \emph{$O$-long} if there is an embedding of (i.e., an order-preserving map from) $O$ into $R$; otherwise $R$ is \emph{$O$-short}.
\end{definition}

In Example~\ref{exa:omega1} we mentioned that $\omega_1$ is the set of all countable ordinals and that ordinals are naturally ordered by inclusion. 
Think of any ordinal $\alpha$ as the linear order of the smaller ordinals, e.g., $\omega_1$ as the linear order of countable ordinals. 

We are particularly interested in $\omega_1$-short ordered sets.
There is a useful positive characterization of such ordered sets.
Recall that sets $X,Y$ of points of an order $\le$ are \emph{cofinal} if for every $x\in X$ there is $y\ge x$ in $Y$ and if for every $y\in Y$ there is $x\ge y$ in $X$
(and \emph{coinitial} is defined symmetrically).

\begin{proposition}\label{pro:from-below}
Let $R$ be an ordered set. The following claims are equivalent.
\begin{enumerate}
\item $R$ is $\omega_1$-short.
\item Every nonempty subset $X$ of $R$ includes a sequence $x_1\le x_2\le\cdots$ cofinal with $X$.
\item Any probability measure $\P$ on $R$ measuring all initial segments $(-\infty,x]$ is continuous from below in the following sense.
For every subset $X$ of $R$, the initial segment $\bigcup_{x\in X} (-\infty,x]$ is measured by $\P$ and
  \begin{equation}\label{eq:LC} 
    \P \left(\bigcup_{x\in X} (-\infty,x]\right) = \sup_{x\in X} \P(-\infty,x].
  \end{equation}
\end{enumerate}
\end{proposition}

\noindent
Here and below, presenting a sequence in the form
\begin{align*}
 x_1 < x_2 < \dots, &\text{ or }x_1\le x_2\le \dots, \text{ or }\\
 x_1 > x_2 > \dots, &\text{ or }x_1 \ge x_2 \ge \dots, 
\end{align*}
we presume that the indices range over the positive integers. Also, we use the convention that the supremum of the empty set of probabilities is zero and the infimum of the empty set of probabilities is one. 

\begin{proof}\mbox{}
\begin{description}
\item[$1\implies2$] We prove the implication $1\implies2$ by contrapositive. Assume that $X$ is a nonempty subset of $R$ such that no sequence $x_1\le x_2\le \cdots$ in $X$ is cofinal with $X$. We construct an embedding of $\omega_1$ into $R$.
Choose $\eta(0)$ arbitrarily in $X$. Suppose that $\beta$ is a countable ordinal and a (possibly transfinite) sequence $\langle \eta(\alpha): \alpha < \beta\rangle$ has been constructed. The sequence contains only countably many elements and thus cannot be cofinal with $X$. Choose $\eta(\beta)$ in $X$ greater than all $\eta(\alpha)$ with $\alpha<\beta$. This way we construct the desired embedding $\langle \eta(\alpha): \alpha < \omega_1\rangle$.

\item[$2\implies3$] Assume 2 and fix an arbitrary subset $X$ of $R$. If $X=\emptyset$ then both sides of Equation~\eqref{eq:LC} are zero. Suppose that $X\ne\emptyset$.  By 2, there is a sequence $x_1\le x_2\le \dots$ of points in $X$ cofinal with $X$.
Accordingly, it suffices to prove that 
$\P(\cup_n (-\infty,x_n]) = \sup_n \P(-\infty,x_n]$,
which follows from the countable additivity of probability measures.

\item[$3\implies1$] Again we prove the desired implication by contrapositive. Assume that there is an embedding $\eta$ of $\omega_1$ into $R$ and let $X = \Range\eta$. We construct a probability measure $\P$ for which Equation~\eqref{eq:LC} fails. Let $\Sigma$ be the p-minimal $\sigma$-algebra of $R$. By Carath\'eodory's extension theorem, to define $\P$ on $\Sigma$, it suffices to define $\P$ on the initial segments $(-\infty,x]$ (in such a way that the conditions of the theorem are satisfied, which we will check carefully later). Define $\P(-\infty,x]=0$ if $X\cap(-\infty,x]$ is countable and $\P(-\infty,x]=1$ otherwise. Then the left side of Equation~\eqref{eq:LC} is 1 while the right side is 0.
It remains to check the applicability of Carath\'eodory's theorem.
To make sure $\P$ is defined on a semi-ring, set $\P(\emptyset)=0$ and $\P(x,y]=\P(-\infty,y]-\P(-\infty,x]$ for all $x<y$.  
Let us first check that $\P(x,y]=0$ if $X\cap(x,y]$ is countable and $\P(x,y]=1$ otherwise.
The case where $X\cap(x,y]$ is countable is trivial, so let us assume that $X\cap(x,y]$ is uncountable.
In this case we have $\P(-\infty,y]=1$ and, since $(x,y]$ contains $\eta(\alpha)$ for a countable ordinal $\alpha$, $\P(-\infty,x]=0$;
by definition, this implies $\P(x,y]=1$.
It remains to check that $\P$ is $\sigma$-additive: if $(x,y]=\cup_{n=1}^{\infty}(x_n,y_n]$, where the union is disjoint, then $\P(x,y]=\sum_{n=1}^{\infty}\P(x_n,y_n]$.
This follows immediately from $X\cap(x_n,y_n]$ being uncountable for at most one $n$.
\qedhere
\end{description}
\end{proof}

\begin{remark}
  The argument in the proof of $1\implies2$ is an instance of definition by transfinite induction (in our case, over the countable ordinals).
  For details, see the transfinite recursion theorem in \cite{Halmos:1960}, Section~18.
  It becomes applicable if we fix a choice function that maps every transfinite sequence $\langle \eta(\alpha): \alpha < \beta\rangle$ in $X$
  for every countable ordinal $\beta$ to $\eta(\beta)\in X$ satisfying our desideratum
  (in this particular case, $\eta(\beta)$ being greater than all $\eta(\alpha)$ with $\alpha<\beta$).
\end{remark}

The reverse of $\omega_1$ is known as $\omega_1^*$. In other words, $\omega_1^*$ is the set of all countable ordinals with the reverse order $\alpha < \beta \iff \beta\in\alpha$. By the below/above symmetry, Proposition~\ref{pro:from-below} has the following corollary.

\begin{corollary}\label{cor:from-above}
Let $R$ be an ordered set. The following two claims are equivalent.
\begin{enumerate}
\item $R$ is $\omega_1^*$-short.
\item Every nonempty subset $X$ of $R$ includes a sequence $x_1\ge x_2\ge\cdots$ coinitial with $X$.
\end{enumerate}
Suppose $R$ is $\omega_1$-short.
Any probability measure $\P$ on $R$ measuring all initial segments $(-\infty,x]$ is continuous from above in the following sense. For any subset $X$ of $R$,
\begin{equation}\label{eq:RC}
 \P\left(\bigcup_{x\in X} [x,\infty)\right) = \sup_{x\in X} \P[x,\infty)
\end{equation}
(which includes the existence of all these probabilities).
\end{corollary}

\begin{proof}
  The only condition that breaks the symmetry is that all initial segments $(-\infty,x]$ should be measurable.
  Therefore, it suffices to check that the measurability of all $(-\infty,x]$ implies the measurability of all $[x,\infty)$,
  i.e., the measurability of all $(-\infty,x)$.
  This follows from $R$ being $\omega_1$-short and part~2 of Proposition~\ref{pro:from-below}.
\end{proof}

Since
\begin{align*}
  1-\P\left(\bigcup_{x\in X} [x,\infty)\right) &= \P\left(\bigcap_{x\in X}(-\infty,x)\right)\\
  1 - \sup_{x\in X} \P[x,\infty) &= \inf_{x\in X} \P(-\infty,x),
\end{align*}
\eqref{eq:RC} is equivalent to 
\begin{equation*}  
 \P\left(\bigcap_{x\in X}(-\infty,x)\right) = \inf_{x\in X} \P(-\infty,x).
\end{equation*}

In our context, the symmetry between $\omega_1$ and $\omega_1^*$ is limited.
$\omega_1^*$ is less dangerous than $\omega_1$.
Replacing the order $\omega_1$ in Example~\ref{exa:omega1} by the reverse order $\omega_1^*$ leads to a p-value that is identically equal to 1,
which is a valid p-value  
(and occurs for very simple test statistics,
e.g., whenever the size of $f$'s codomain is 1) albeit not useful. Still, $\omega_1^*$-short orders have desirable properties. Corollary~\ref{cor:from-above} indicates some of them. Others will will be pointed out later. The following definition is borrowed from the literature on linear orders; see \cite{G29} and \cite[p.~88]{Rosenstein}.
  
\begin{definition}  
  A linear order is \emph{short} if it is $\omega_1$-short and $\omega_1^*$-short.
\end{definition}

Now we are ready to introduce our general notion of test statistic.

\begin{definition}\label{def:stat}
  A \emph{test statistic} is a nominal test statistic whose codomain is short.
\end{definition}

\begin{proposition}
Any ordered set with second-countable order topology is short. Thus every nearly traditional test statistic is a test statistic.
\end{proposition}

\noindent
In particular, the real line $\R$ is short.

\begin{proof}
  Let $R$ be an ordered set with second-countable order topology.
  By Proposition~\ref{prop:Cater}, $R$ is separable and has at most countably many jumps.
  Let $C$ be a countable set that is dense in $R$ and contains all points involved in jumps.
  Suppose toward a contradiction that $\eta$ is an order-preserving or order-reversing map from $\omega_1$ to $R$.
  There is a point $c_\alpha\in C$ in any open interval between $\eta(\alpha)$ and $\eta(\alpha+2)$.
  We obtain uncountably many distinct points $c_\alpha$, which is impossible.
\end{proof}

\section{Induced test statistics, p-functions, and p-values}

\subsection{Induced test statistics}

\begin{definition}
Any test statistic $f$ for a probability trial $(\Omega,\Sigma,\P)$ \emph{induces} a traditional test statistic $\f(x) = \P[f\le f(x)]$ on $(\Omega,\Sigma,\P)$.
\end{definition}

\begin{lemma}\label{lem:self-induced}
If $f$ is any test statistic then $\hathat{f}=\f$. In other words, any induced test statistic is self-induced. 
\end{lemma}

\begin{proof} 
Let $f$ be a test statistic for a probability trial $\T = (\Omega,\Sigma,\P)$ and $R$ be the codomain of $f$. By Definition~\ref{def:induced-prob}, $f$ induces a probability distribution $\P_f(X) = \P(f\1(X))$ on $R$. To simplify notation, we omit the subscript $f$.

Let $\eta(r) = \P[f\le r]$, so that $\eta(f(x)) = \P[f\le f(x)] = \f(x)$. If $s\le r$ then $\eta(s)\le\eta(r)$. On the other hand, 
\begin{align*}
\eta(s)\le\eta(r) 
&\iff\text{either $s\le r$ or else ($s>r$ and $\eta(s)=\eta(r)$})\\
&\iff\text{either $s\le r$ or else ($s>r$ and $\P(r,s]=0$}).
\end{align*} 

For any $r\in R$, the set $S_r = \{s\in R: s>r \text{ and }\P(r,s]=0\}$ is measurable in $R$ and $\P(S_r)=0$. Indeed, by the definition of nominal test statistics, the initial segments $(-\infty,t]$ of $R$ are measurable. Since $f$ is a genuine test statistic (rather than just nominal), $R$ is $\omega_1$-short and so there is a sequence $s_1 \le s_2 \le \dots$ in $S_r$ cofinal with $S_r$ so that $S_r = \bigcup_n (r,s_n]$ and thus is measurable. Further, $\P(\cup_n (r,s_n]) = \lim_n \P(r,s_n] = 0$. Thus
\[
0 = \P(S_r) = \P(f^{-1}(S_r)) =
\P\{y: f(y)>r\text{ and }\eta(f(y))=\eta(r)\}.
\]
Now we are ready to prove $\hathat{f}(x) = \f(x)$. 
\begin{align*}
\hathat{f}(x) &= \P[\f\le\f(x)] = \P\{y: \eta(f(y))\le\eta(f(x))\}\\ 
&= \P\{y: f(y)\le f(x)\} + \P\{y: f(y)>f(x)\text{ and }\eta(f(y))=\eta(f(x))\}\\
& = \P\{y: f(y)\le f(x)\} = \P[f\le f(x)] = \f(x).
\qedhere
\end{align*}
\end{proof}

\begin{theorem} \label{thm:range-exact}
Let $f$ be a traditional test statistic with values in the real segment $[0,1]$. The following claims are equivalent.
\begin{enumerate}
\item $f$ is induced by some test statistic.
\item $f$ is self-inducing, i.e., $\f=f$.
\item $\P[f\le\eps] = \eps$ for every $\eps\in\Range f$.
\end{enumerate}
\end{theorem}

\begin{proof}
Obviously 2 implies 1. By Lemma~\ref{lem:self-induced}, 1 implies 2. It suffices to prove that 2 and 3 are equivalent.
\begin{description}
\item[$2\implies3$] Assume 2 and let $\eps = f(x)$ for some $x$. We have
\[\P[f\le f(x)] = \f(x) = f(x) = \eps.\]
\item[$3\implies2$] Assume 3. Given a sample point $x$, let $\eps = f(x)$. We have
\[
\f(x) = \P[f\le f(x)] = \P[f\le \eps] = \eps = f(x).
\qedhere
\]
\end{description}
\end{proof}

The definition of the induced test statistic $\f$ can be extended to the case where $f$ is only a nominal test statistic. But the following proposition emphasizes the important role of the property of being $\omega_1$-short.

\begin{proposition}\label{thm:short} 
Let $R$ be an $\omega_1$-long ordered set. There exists an $R$-valued nominal test statistic $f$ such that $\P[\f\le0] = 1$ even though $0\in\Range\f$. 
\end{proposition}

\begin{proof}
  Fix an embedding $\eta$ of $\omega_1$ into $R$, and let \[L =\{x\in R: x\le \eta(\alpha)\text{ for some }\alpha\in\omega_1\}.\] 
  Let $\Sigma$ be the least $\sigma$-algebra on $R$ that contains the initial segments $(-\infty,x]$ and also contains $L$.  

  For every member $X$ of $\Sigma$, either $X$ or $R-X$ contains at most a countable subset of $\Range\eta$.
  Indeed, every initial segment of the form $(-\infty,x]$ and $L$ have this property, and the property is preserved by complementation and countable unions. 

  Define a probability measure $\P$ on $\Sigma$ as follows:
  If $X$ contains at most a countable subset of $\Range\eta$ then $\P(X)=0$; otherwise $\P(X)=1$.

  The desired nominal test statistic $f$ is the identity function on $R$.
  It is easy to see that $x\in L$ if and only if $\f(x) = \P[f\le f(x)] = 0$. So $\P[\f\le0] = \P(L) = 1$. 
\end{proof}

\subsection{p-functions and p-values}

We want to define p-functions and p-values in such a way that p-values are the values of p-functions. It is tempting to define a p-function as a traditional test statistic $\f$ induced by some test statistic $f$.
By Theorem~\ref{thm:range-exact}, we have $\P[f\le\eps]=\eps$ for any $\eps\in\Range f$. But in practice people also use conservative p-values. To accommodate this practice, we give a more general definition of p-functions. 

\begin{definition}
\mbox{}
\begin{itemize}
\item A \emph{p-function} is a traditional test statistic $f$ with values in the real segment $[0,1]$ such that $\P[f\le\eps]\le\eps$ for every $\eps\in[0,1]$.
\item A p-function $f$ is \emph{range-exact} if $\P[f\le\eps] = \eps$ for every $\eps\in\Range f$; otherwise $f$ is \emph{conservative}.
\item A p-function $f$ is \emph{everywhere exact} or simply \emph{exact} if $\P[f\le\eps] = \eps$ for every $\eps\in[0,1]$.
  \end{itemize}
\end{definition}

If $f$ is a p-function then $cf$ is a p-function for every $c\ge1$. Indeed
\[
 \P[cf\le\eps] = \P[f\le\eps/c] \le
 \eps/c
 \le \eps.
\]
If $c\in(0,1)$ then $cf$ may not be a p-function.
In particular if $\P[f\le\eps] = \eps$ for at least one $\eps>0$
then $cf$ is not a p-function because, for that $\eps$, we have
\[\P[cf\le c\eps] = \P[f\le \eps] = \eps > c\eps.\]

\begin{theorem}\label{thm:pfunction}
Let $f$ be a traditional test statistic with values in the real segment $[0,1]$. The following claims are equivalent.
\begin{enumerate}
\item $f$ is an induced test statistic.
\item $f$ is a range exact p-function.
\end{enumerate}
\end{theorem}

\begin{proof}
2 implies 1 by Theorem~\ref{thm:range-exact}. To prove the other implication, assume 1. By Theorem~\ref{thm:range-exact}, $\P[f\le\eps]=\eps$ for any $\eps\in\Range f$. It remains to prove that $\P[f\le\eps]\le\eps$ for every $\eps\in[0,1]-\Range f$. 

Let $\eps\in[0,1]-\Range f$, $S = \{s\in\Range f: s<\eps\}$ and $\eps_0 =\sup S$.
  If $S=\emptyset$ then $[f\le\eps]=\emptyset$ and $\P[f\le\eps] = 0\le\eps$. Otherwise there is a sequence $s_1 \le s_2 \le \cdots$ of reals in $S$ converging to $\eps_0$. Then
  \begin{multline*}
    \P[f\le\eps] = \P[f\le\eps_0]
    =
    \P
    \left(
      \bigcup_{n=1}^{\infty} [f\le s_n]
    \right)\\
    =
    \lim_{n\to\infty}
    \P
    \left(
      [f\le s_n]
    \right)
    =
    \lim_{n\to\infty}
    s_n
    =
    \eps_0 \le \eps.
    \qedhere
  \end{multline*}
\end{proof}

\begin{definition}
Let $F$ be a p-function for some probability trial $\T = (\Omega,\Sigma,\P)$. For any outcome $x\in\Omega$, the number $F(x)$ is the \emph{p-value} associated with the test statistic $F$ and the outcome $x$. If $F$ is an induced test statistic and $f$ is any test statistic for $T$ inducing $F$ then $F(x)$ is also the p-value associated with $f$ and $x$. If $\P[F\le F(x)] = F(x)$ then the p-value $F(x)$ is \emph{exact}; otherwise it is \emph{conservative}.
\end{definition}

\section{Diffuse test statistics and exact p-functions}

We are particularly interested in exact p-functions $f(x)$, the p-functions with $\P[f\le\eps]=\eps$ for all $\eps\in[0,1]$.
Classical parametric statistics is a rich source of exact p-functions; randomized p-values (discussed in the next section) is another important example.

Recall that an \emph{atom} in a probability space $\T$ is an event (i.e., a measurable set) of positive probability that cannot be split into a disjoint union of two events of positive probability. $\T$ is \emph{diffuse}
if it has no atoms. Every singleton event of positive probability is an atom. In Example~\ref{exa:omega1}, the complement of every countable set is an atom. 

\begin{lemma}\label{lem:diffuse}
Let $R$ be a short ordered set as well as a probability space where all initial segments $(-\infty,r]$ are measurable. $R$ is diffuse if and only if it has no singleton atoms. 
\end{lemma}

\begin{proof}
The ``only if'' implication is trivial. To prove the ``if'' implication,
suppose toward a contradiction that $R$ has no singleton atoms and yet it does have an atom $A$. The subset $A$ by itself is a short ordered set. If $\Sigma$ is the $\sigma$-algebra of $R$, and $\P$ is the probability measure on $\Sigma$, consider the $\sigma$-algebra $\Sigma_A = \{A\cap X: X\in\Sigma\}$ and restrict $\P$ to $\Sigma_A$. In the rest of the proof, we work with $A$. Without loss of generality, we may assume that $A$ is the whole $R$. Accordingly $R$ itself is an atom in $R$.

Obviously $\P(R) = 1$. Since $R$ is an atom, for any $X\in\Sigma$, the probability $\P(X)$ is either 0 or 1. Notice that every singleton set $\{r\}$ is measurable in $R$. Indeed, $(-\infty,r]$ is measurable, so it suffices to prove that $(-\infty,r)$ is measurable. By Proposition~\ref{pro:from-below}, there is a sequence $x_1\le x_2 \le\dots$ converging to $r$, and so $(-\infty,r)$ is a countable union of measurable sets.

Let $I$ be the set of points $x\in R$ with $\P(-\infty,x]=0$ (it is an initial segment, in the sense of containing any $y$ such that $y\le x$ for some $x\in I$). By the continuity from below, see Equation~\eqref{eq:LC},
\[
 \P(I) = \P \left(\bigcup_{x\in I} (-\infty,x]\right)
       = \sup_{x\in I} \P(-\infty,x] = 0.
\]
Let $F$ be the final segment $R-I$. For any $y\in F$, $\P(-\infty,y]=1$ and therefore $\P(y,\infty)=0$. Since singleton sets are measurable in $R$ and $R$ has no singleton atoms, every $\P[y,\infty) = \P\{y\} + \P(y,\infty)=0$. By the continuity from above, see Equation~\eqref{eq:RC},
\[
 \P(F) = \P \left(\bigcup_{x\in F} [x,\infty)\right)
       = \sup_{x\in F} \P[x,\infty) = 0.
\]
Thus $\P(R) = \P(I) + \P(F) = 0$, which gives us the desired contradiction. 
\end{proof}

Call a nominal test statistic $f$ on a probability trial $\T = (\Omega,\Sigma,\P)$ \emph{diffuse} if the probability distribution $\P_f$ that $f$ induces on its codomain $R$ is diffuse. If $R$ is short then, by Lemma~\ref{lem:diffuse}, the test statistic $f$ is diffuse if and only if $\P(f\1(r))= \P_f(r) = 0$ for every point $r\in R$. 

\begin{proposition}
  \mbox{}
  Let $f$ be a diffuse test statistic.
  Then the induced test statistic $\f$ is diffuse.
\end{proposition}

\begin{proof}
The codomain $[0,1]$ of $\f$ is short, and every singleton set in $[0,1]$ is measurable. By Lemma~\ref{lem:diffuse} it suffices to prove that $\P[\f=\eps]=0$ for every $\eps\in\Range\f$. Fix such a number $\eps$. Since $\eps\in\Range\f$, the set $X=\{x: \f(x)=\eps\}\ne\emptyset$. Let $R_0 = \{f(x): x\in X\}$.

We use the notation and results established in the proof of Lemma~\ref{lem:self-induced}. If $x\in X$ and $r=f(x)$, we have $\eps = \f(x) =  \P[f\le r] = \eta(r)$ and so $[\f=\eps] = \{y: \f(y) = \f(x)\} = \{y: \eta(f(y)) = \eta(r)\}$. Let $U(r) = \{y: f(y)\ge r\text{ and }\eta(f(y)) = \eta(r)\}$ and $U'(r) = \{y: f(y)> r\text{ and }\eta(f(y)) = \eta(r)\}$. In the proof of Lemma~\ref{lem:self-induced} we established that $\P(U'(r))=0$. But $U(r) = U'(r) \cup f\1(r)$. Since $f$ is diffuse, $\P(f\1(r))= \P_f(r) = 0$ and so $\P(U(r))=0$. 

If $r=\min R_0$ then $\P[\f=\eps] = \P[\f=\eta(r)] = \P(U(r)) = 0$. Suppose that $R_0$ does not have a minimal element. Since $R$ is short, there exists a sequence $r_1 > r_2 > \dots$ in $R_0$ that is coinitial with $R_0$. We have
\[
 \P[\f=\eps] = \P\left(\bigcup_n U(r_n)\right) = 0.
 \qedhere
\]
\end{proof}

\begin{theorem}\label{thm:atomless}\mbox{}
Let $f$ be a diffuse test statistic with a short codomain. 
Then the induced test statistic $\f$ is an exact p-function.
\end{theorem}

\begin{proof}
By Theorem~\ref{thm:pfunction}, $\f$ is a range-exact p-function, so that $\P[f\le\eps]\le\eps$ for $\eps\in[0,1]$ and $\P[f\le\eps]=\eps$ for all $\eps\in\Range\f$. It remains to prove that $\P[f\le\eps]=\eps$ for every $\eps\in[0,1]-\Range\f$. Fix such a number $\eps$.

Let $\eps_0 = \sup\{\delta\in\Range\f: \delta < \eps\}$ and $\eps_1 = \inf\{\delta\in\Range\f: \delta > \eps\}$. Our convention here is that $\sup\emptyset=0$ and $\inf\emptyset=1$. There exists a sequence $\delta_1\le \delta_2\le \dots$ in $\Range\f$ that converges to $\eps_0$, so that
\[
 \P[\f\le\eps_0] = \lim_n\P[\f\le\delta_n] + \P[\f=\eps_0] = 
 \lim_n \delta_n + 0 = \eps_0.
\]
Similarly, there is a sequence $\delta_1\ge \delta_2\ge \dots$ in $\Range\f$ that converges to $\eps_1$, so that
\[
 \P[\f\ge\eps_1] = \P[\f=\eps_1] + \lim_n\P[\f>\delta_n] = \lim_n (1-\delta_n) = 1 - \eps_1.
\]
We have
\[
  1 
  = \P[\f\le\eps_0] + \P[\f\ge\eps_1] = \eps_0 + (1-\eps_1),
\]
so that $\eps_0 = \eps_1 = \eps$ and $\P[\f\le\eps] = \P[\f\le\eps_0] = \eps_0 = \eps$.
\end{proof}

\begin{proposition}
  If $R$ is an $\omega_1^*$-long (and $\omega_1$-short) linearly ordered set endowed with the p-minimal $\sigma$-algebra
  then there is an $R$-valued diffuse nominal test statistic $f$ such that $\hat f$ is not an exact p-function.
\end{proposition}

\begin{proof}
Consider a trial $\T=(\Omega,\Sigma,\P)$ where $\Omega,\Sigma,\P$ are as follows.
\begin{itemize}
\item $\Omega$ is $\omega_1^*$, i.e., the set of countable ordinals with the reverse order $\alpha < \beta \iff \beta \in \alpha$. 
\item $\Sigma$ is the p-minimal $\sigma$-algebra on $\Omega$. It consists of the countable subsets of $\Omega$ and their complements. 
\item $\P(X) = 0$ if $X$ is countable. In particular, every $\P(-\infty,\alpha]=1$.
\end{itemize} 

Since $R$ is $\omega_1^*$-long, there is an order isomorphism $\eta$ from $\omega_1^*$ into $R$;
$\eta$ is an $R$-valued nominal test statistic on $\T$. Indeed, since the $\sigma$-algebra of $R$ is p-minimal, it suffices to show that every $\eta^{-1}(-\infty,r]\in\Sigma$. Since $\omega_1$ is well-ordered, every $X\subseteq\omega_1$ has a minimal point in $\omega_1$; accordingly every $X\subseteq\omega_1^*$ has a maximal point in $\omega_1^*$.
  In particular, let $y = \max\{\alpha\in\omega_1^*: \eta(\alpha)\le r\}$. Accordingly, $\eta^{-1}(-\infty,r] = \eta^{-1}(-\infty,\eta(y)] = (-\infty,y]\in\Sigma$. 

  Since $\P$ takes only two values, the induced p-function $\hat{\eta}$ takes only two values and thus is not exact.
  \qedhere
\end{proof}

\section{Randomized p-values}
\label{sec:randomized}

In this section we discuss randomized p-values as a natural application of non-traditional test statistics.
Randomized p-values arise naturally in situations where the distribution of the test statistic is not continuous.
They are produced by test statistics whose codomain is $\R\times[0,1]$ with the lexicographic order;
intuitively, we add a random number to a traditional test statistic to break ties if there are any.
This makes the distribution of a randomized p-value uniform over the segment $[0,1]$
(as shown in Theorem~\ref{thm:randomized} below).

  As discussed in Appendix~A, Egon Pearson \cite{Pearson:1950} defended randomized p-values in 1950,
  but he was mainly writing about the abstract notion.
At this time there are at least two (and probably many more) fields of 
statistics where randomized p-values are essential:
multiple hypothesis testing in bioinformatics 
(see, e.g., \cite{Dickhaus})
and on-line testing the hypothesis of exchangeability using conformal martingales
(see, e.g., \cite[Section~7.1]{Vovk/etal:2005}).
In both cases p-values are used repeatedly a large number of times,
and any non-uniformity of their distribution quickly accumulates and destroys the power of the overall procedure.

  Recall that the product $\P_1\times\P_2$ of probability measures $\P_1$ and $\P_2$
  on measurable spaces $(\Omega_1,\Sigma_1)$ and $(\Omega_2,\Sigma_2)$ respectively
  is the unique probability measure on $(\Omega_1\times\Omega_2,\Sigma_1\otimes\Sigma_2)$
  with $\P(X_1\times X_2) = \P_1(X_1) \cdot \P_2(X_2)$ for $X_1\in\Sigma_1$ and $X_2\in\Sigma_2$.
We are essentially in the situation of Example~\ref{ex:lex}, except that the assumption that the codomain of $f$ is countable is dropped.

By Theorem~\ref{thm:range-exact}, the induced p-function $f$ of a test statistic is only guaranteed to satisfy the inequality $\P[f\le\eps]\le\eps$,
and very simple examples show that $\P[f\le\eps]<\eps$ is indeed possible:
e.g., take any $\epsilon\in(0,1)$ when the sample space is a singleton.
Randomized p-values are a way of making $f$ exact, i.e., achieving the equality $\P[f\le\eps]=\eps$ for all $\epsilon\in[0,1]$.
Informally, we enrich our probability space by adding a random number generator and using its output for breaking ties
in values of the test statistic for different outcomes.

We need a couple of auxiliary results.

\begin{lemma}\label{lem:lex}
  If $R$ and $S$ are short orders,
  then the product $R\times S$, ordered lexicographically,
  is short.
\end{lemma}

\begin{proof}
By symmetry it suffices to prove that $R\times S$ is $\omega_1$-short, i.e., that for every well-ordered set $A$, if there exists an order preserving map $\eta: A\to R\times S$, then $A$ is countable. Each $\eta(a)$ has the form $(r_a,s_a)$, and $(r_a,s_a)< (r_b,s_b)$ if and only if either $r_a<r_b$ or else both $r_a=r_b$ and $s_a<s_b$.

Since $S$ is short, the set 
$A_r = \{a: r_a = r\}$ is countable for every $r\in R$. Since $A$ is well-ordered, the subset $\{\min A_r: A_r\ne\emptyset\}$ of $A$ is well ordered. Since $R$ is short, the well-ordered subset $\{\eta(\min A_r): A_r\ne\emptyset\}$ of $R$ is countable, so that there are only countable many nonempty sets $A_r$. Therefore $A$ is a countable union of countable sets, so that $A$ is countable.
\end{proof}

\begin{lemma}\label{lem:product}
Let $(\Omega_1,\Sigma_1,\P_1)$ and $(\Omega_2,\Sigma_2,\P_2)$ be probability spaces where every singleton set is measurable, and form the product probability space
\[(\Omega_1\times\Omega_2,\Sigma_1\otimes\Sigma_2,\P_1\times\P_2).\]
Every singleton set is measurable in the product space. Furthermore, the product space has the property that the probability of every singleton event is zero if at least one of the factors has the property.
\end{lemma}

\begin{proof}
By the definition of the product, 
$\P(X_1\times X_2) = \P_1(X_1)\cdot \P_2(X_2)$
for any $X_1\in\Sigma_1$ and $X_2\in\Sigma_2$. A singleton set in $\Omega$ has the form $\{x_1\}\times\{x_2\}$ and this is measurable. The second claim follows from the fact that $\P(\{x_1\}\times\{x_2\}) = \P_1\{x_1\} \cdot \P_2\{x_2\}$.
\end{proof}

Now we are ready to address the issue of randomized p-values. We start from our usual setting of a given traditional test statistic $f:\Omega\to\R$
on a trial $\T = (\Omega,\Sigma,\P)$.
The output of a random number generator is modelled as the trial $([0,1],\B,\U)$,
where $\B$ is the Borel $\sigma$-algebra on $[0,1]$ and $\U$ is the uniform probability measure on $([0,1],\B)$.
The overall trial is now the product
\[
  \bar \T
  =
  (\bar\Omega,\bar\Sigma,\bar\P)
  =
  (\Omega\times[0,1],\Sigma\otimes\B,\P\times\U)
\]
and the test statistic $f$ on $\Omega$ is replaced by a finer test statistic
\begin{equation}\label{eq:finer}
  F(x,r)
  =
  (f(x),r)
\end{equation}
on $\bar\Omega$.
The order on the codomain $\R\times[0,1]$ of $F$ is lexicographic,
\[
  (p,r)\le (q,s) \iff p<q \text{ or }
  (p=q \text{ and } r\le s).
\]
Intuitively, this means that the impugning power of our test statistic is determined by $f$,
and the outcome of the random number generator is only used for tie breaking.
Let us call all functions $F$ that can be obtained in this way
\emph{randomized traditional test statistics}.
They have the following useful property.

\begin{theorem}\label{thm:randomized}
  The induced p-function $\hat F$ of any randomized traditional test statistic $F$ is exact, so that $\bar\P[\hat F\le\eps]=\eps$ for any $\epsilon\in[0,1]$.
\end{theorem}

\begin{proof}
By Lemma~\ref{lem:lex}, the codomain $\R\times[0,1]$ of $F$ is short. It is easy to see that all initial segments $(-\infty,x]$ of $\R\times[0,1]$ are measurable. By Lemma~\ref{lem:diffuse}, the $\P_F$-atoms, if any, of $\R\times[0,1]$ are singletons. By Lemma~\ref{lem:product} and $\P_F=\P_f\times\U$ (cf.\ \eqref{eq:finer}), $\P_F$ does not have singleton atoms and thus is diffuse. It remains to apply Theorem~\ref{thm:atomless}.
\end{proof}

It is easy to see that the theorem generalizes to the case where the component $f$ of test statistic $F$ is any test statistic.

The value $\hat F(x,r)$ is the randomized p-value corresponding to an outcome $x$ and random number $r$.
But it is easy to see that the topology of the lexicographically ordered $\R\times[0,1]$ is not second-countable (not even separable).
According to Definition~\ref{def:nearly} and Proposition~\ref{prop:Cater},
the function $F$ defined by \eqref{eq:finer} is not a traditional or even nearly-traditional test statistic.
But $\R$ and $[0,1]$ are short, and so $F$ is a test statistic according to Definition~\ref{def:stat},
as Lemma~\ref{lem:lex} shows.

\begin{remark}
  When using randomized p-values, statisticians (and computer scientists in related areas) do not usually emphasize
  the use of non-traditional test statistics, which remain implicit.
  They prefer to define randomized p-values from scratch rather than using the standard definition \eqref{eq:p}.
  Namely, the usual definition (as given in, e.g., \cite{Dickhaus} and \cite{Vovk/etal:2005}) is
  \begin{equation}\label{eq:usual}
    \hat F(x,r)
    =
    \P[f<f(x)]
    +
    r
    \P[f=f(x)],
  \end{equation}
  where $r$ is a random number in $[0,1]$.
  The only explicit use of the lexicographic order in connection with randomized p-values
  that we are aware of is in \cite[p.~91, (3.4)]{Coudin:2007}.
\end{remark}

\section{An alternative to randomizing p-values}

Many practical statisticians dislike the idea of randomized p-values.
A pioneer of randomized p-values (Stevens on the last page of \cite{Stevens:1950}) says:
``We suppose that most people will find repugnant the idea of adding yet another random element
to a result which is already subject to the errors of random sampling.''
When reporting on Anscombe's previous work \cite{Anscombe:1948}
he says that the method ``was there dismissed rather briefly as being unsatisfactory''.
Egon Pearson \cite{Pearson:1950} was more positive but still admitted that there are
``a number of objections'' to the use of the method,
``which many statisticians would regard as decisive.''

A popular alternative to randomized p-values is mid-p-values,
introduced by Lancaster in 1961 \cite{Lancaster:1961},
his motivation being that in some cases computing randomized p-values may be
``time-consuming and even embarrassing to the statistician.''
The \emph{mid-p-value} is defined to be the following modification of~\eqref{eq:usual}:
\begin{equation}\label{eq:mid-p-value}
  \P[f<f(x)]
  +
  \frac12
  \P[f=f(x)];
\end{equation}
in other words it is defined to be the arithmetic mean of $\P[f<f(x)]$ and $\P[f\le f(x)]$
(whereas the randomized p-value is distributed uniformly between $\P[f<f(x)]$ and $\P[f\le f(x)]$).
The corresponding \emph{mid-p-function} maps each outcome $x\in\Omega$ to the mid-p-value \eqref{eq:mid-p-value}.
The definition of mid-p-values is natural,
but the main problem with it is that mid-p-functions are not guaranteed to be p-functions
(and they are not p-functions in interesting cases).

In randomized p-values we complement a given test statistic by a random number to break ties.
We can easily imagine less repugnant (to use Stevens's expression) ways of tie-breaking
using the lexicographic order on $\R^2$ (or $\R^k$ for $k>2$).
In this section we will discuss several specific examples,
but it will be clear that the approach is general.
Although mathematically less elegant than randomizing p-values,
it reduces the conservativeness of p-functions while maintaining their validity.

Wilcoxon's rank-sum test \cite{Wilcoxon:1945,Deuchler:1914} is the workhorse of nonparametric statistics.
It is used for comparing two groups of observations (real numbers),
  \[
    x_1,\ldots,x_m \text{ and } y_1,\ldots,y_n;
  \]
for simplicity we will assume that these observations are all different.
The test statistic $R_x$ is the sum of the ranks of the first $m$ observations, in this case the sum of the ranks of $x_1,\ldots,x_m$,
where all $m+n$ observations are ranked from 1 (the smallest) to $m+n$ (the largest).
The p-value corresponding to the given value $R_x$ of the test statistic is the probability that $R'_x\le R_x$.
Here $R'_x$ is obtained by applying the test statistic to a random permutation $z_1,\ldots,z_{m+n}$
of the $m+n$ observations; in other words $R'_x$ is the sum of the ranks of observations $z_1,\dots,z_m$.
Intuitively, we are testing the null hypothesis that all $m+n$ observations are drawn independently
from the same continuous distribution on the real line
against the alternative that the $x$s tend to be smaller than the $y$s.

Wilcoxon's rank-sum test is remarkably efficient
(see, e.g., \cite[Section II.4]{Lehmann:2006}),
but its p-values have a discrete distribution;
first of all, it is clear that this distribution is concentrated on the set
\[
  \left\{
    \frac{1}{\binom{m+n}{m}},
    \frac{2}{\binom{m+n}{m}},
    \ldots,1
  \right\},
\]
as for any permutation test.
For example, when $m=n=6$, it is concentrated on the set $\{1/924,2/924,\ldots,1\}$.
However, Table~5.1 in \cite[Chapter~5]{Pratt/Gibbons:1981} shows that Wilcoxon's test statistic
is much cruder: it takes values in the set
\[
  \{1/924, 2/924, 4/924, 7/924, 12/924, 19/924, 30/924, 43/924, 61/924,\ldots,1\}
\]
(where the ellipsis does not imply that the reader is supposed to be able to fill in the missing values).

To partially break the ties between the values of the test statistic
(this is the only kind of ties we are interested in since all observations were assumed to be different),
we can complement $R_x$ by the value $T_2$ of the Fisher--Yates--Terry statistic,
which is computed similarly to $R_x$ but applies a monotonic transformation to all the ranks changing the sums of ranks accordingly.
The new test statistic, $(R_x,T_2)$, takes values in $\R^2$ equipped with the lexicographic order.
The p-value corresponding to the given value of $(R_x,T_2)$
is the probability that $(R'_x,T'_2)\le(R_x,T_2)$,
where $(R'_x,T'_2)$ is computed by applying the test statistic
to a random permutation of our $m+n$ observations.
The range of the new p-function will be extended by adding the points in
\begin{multline*}
  \{5/924, 8/924, 10/924, 14/924, 15/924, 17/924, 21/924,\\
  22/924, 24/924, 26/924, 28/924, 32/924, 34/924, 35/924,\\
  37/924, 39/924, 40/924, 42/924, 48/924, 49/924,\ldots\},
\end{multline*}
as the same table in \cite{Pratt/Gibbons:1981} shows
(although we suspect that Pratt and Gibbons's results may be affected by the limited numeric accuracy of their calculations).

Lots of numbers of the form $k/924$, $k=1,2,\ldots$, are still missing,
so we might add the van der Waerden statistic $T_3$,
which is similar to $T_2$ but uses a slightly different monotonic transformation.
Now the combined test statistic $(R_x,T_2,T_3)$ takes values in $\R^3$ with the lexicographic order.
However, adding $T_3$ will only add one number, $41/924$,
to the intersection of the range of the p-function and $[0,49/924]$.
The reason for this poor tie-breaking performance of $T_3$
is that the test statistics $T_2$ and $T_3$ are so similar:
both are based on monotonic transformations of ranks defined in terms of the Gaussian distribution.
One way to break the ties more efficiently is to replace $T_3$
by a test statistic analogous to $T_2$ or $T_3$ but based on, e.g.,
monotonic transformations of ranks defined in terms of the Laplace distribution
(popular in robust statistics).

An advantage of all these non-traditional test statistics is that the corresponding p-function will be valid
(in the sense of \eqref{eq:valid})
whenever the observations are generated independently from the same continuous probability distributions on $\R$,
and no parametric assumptions are required.
A disadvantage is that even in the ideal situation (from the point of view of breaking ties)
the range of the test statistic is $\{1/924,2/924,\ldots,1\}$
(where we revert to the normal use of the ellipsis: the reader is expected to fill it in).
Therefore, the distribution is still not uniform on $[0,1]$,
although it is uniform on $\{1/924,2/924,\ldots,1\}$ (and so ``almost uniform'' on $[0,1]$).
In the rest of this section we will assume that we are in this ideal situation.

A drastic step perfectly breaking all ties (with probability one)
but partly sacrificing the non-parametric character of the test
is to add Student's \cite{Gosset:1908} $t$-statistic
\[
  t
  =
  \frac{\bar x - \bar y}{S},
\]
where we ignore an irrelevant constant factor and use the notation
\begin{align*}
  \bar x &= \frac1m \sum_{i=1}^m x_i, \qquad \bar y = \frac1n \sum_{i=1}^n y_i,\\
  S &= \sqrt{\sum_{i=1}^m(x_i-\bar x)^2 + \sum_{i=1}^n(y_i-\bar y)^2},
\end{align*}
to the list $(R_x,T_2,\ldots)$,
with the order on $(R_x,T_2,\ldots,t)$ still being lexicographic.
The p-value corresponding to a given value of $(R_x,T_2,\ldots,t)$
is the probability that $(R'_x,T'_2,\ldots,t')\le(R_x,T_2,\ldots,t)$,
where $(R'_x,T'_2,\ldots,t')$ is computed by applying the test statistic to a random sample of size $m+n$
drawn independently from the standard Gaussian distribution.
Because of the nonparametric nature of the test statistics preceding $t$,
the resulting p-function $f$ will satisfy $\P[f\le\eps]=\eps$
for any $\eps\in\{1/924,2/924,\ldots,1\}$ and the power $\P=P^{m+n}$
of any continuous probability distribution $P$ on $\R$.
On the other hand, we will have $\P[f\le\eps]=\eps$
for any $\eps\in[0,1]$ and the power $\P=P^{m+n}$ of any Gaussian distribution $P$ on the real line.

\begin{remark}
  In principle, we could have used only traditional test statistics in this section
  since even the most complicated of our test statistics,
  $(R_x,T_2,\ldots,t)$ used in the last paragraph,
  had all components but one taking values in discrete sets.
  As we know, such orders can be embedded in the real line.
  However, the resulting traditional test statistic would be awkward,
  and it is much more natural to think in terms of the original test statistic,
  such as $(R_x,T_2,\ldots,t)$
  (cf.\ Example~\ref{ex:lex}).
  And even for a practical statistician,
  it may be reassuring to know that she is on safe ground when using any $\R^k$-valued test statistics
  (with the lexicographic order on $\R^k$).
\end{remark}

\section{A summary}

This paper's aim has been to investigate advantages and drawbacks of various classes of nominal test statistics.
If forced to choose one class,
our recommendation would be to use the class of test statistics
(i.e., nominal test statistics with short codomains).
In view of Theorem~\ref{thm:range-exact} and Proposition~\ref{thm:short},
this will lead to valid p-values (but possibly conservative p-functions).
By Theorem~\ref{thm:atomless},
the corresponding p-functions will be exact in the case of diffuse test statistics.
Finally, Theorem~\ref{thm:randomized} is applicable to any test statistic
(as we say after its proof)
and allows us to define exact p-functions by using the device of randomization.

\subsection*{Acknowledgments}

We thank Andreas Blass, Steffen Lauritzen, and Glenn Shafer for useful comments on drafts of this article.

  \appendix
  \section{History of p-values}

  In the applications of statistics,
  the results of hypothesis testing are almost invariably packaged as p-values.
  However, it's difficult to pinpoint when exactly the term was introduced;
  it developed slowly and informally.
  The expression ``p-value'' has a whiff of barbarity about it,
  and there are at least a dozen ($3\times2^2$) different ways of writing it.
  Do you capitalize ``p''?
  (If not in general, do you capitalize it at the start of a sentence?)
  Do you set it in italics?  Do you put a hyphen after it?
  We will ignore such differences in this discussion.
  But even then, the list of alternative expressions is bewildering.
  Herbert A. David's list \cite[p.~211]{David:2001} includes 9 such expressions:
  probability level, sample level of significance, observed significance level, significance probability,
  descriptive level of significance, critical level, significance level, prob-value, and associated probability.
  And that list is clearly incomplete: e.g., ``achieved significance level'' (usually abbreviated to ASL)
  and ``attained significance level'' are also popular.

  The situation with the notion (rather than the term) of p-value is different,
  and it has a long and venerable history.
  To our knowledge, the first p-value was computed by John Arbuthnott in 1710 \cite{Arbuthnott:1710};
  having observed that the number of male births in London exceeded the number of female births
  during each of the 82 years from 1629 to 1710,
  he rejected the null hypothesis of even chances for the birth of male or female.
  He evaluated the p-value, which he referred to using the generic term ``lot'', as $2^{-82}$.
  This was, however, a very special case as the value attained by his chosen test statistic
  (the number of years in which more males were born)
  was extreme (82 out of 82).
  The p-value was one-sided, which was justified because, in Arbuthnott's words,
  ``the external accidents to which are males subject (who must seek their food with danger)
  do make a great havock of them.''

  As far as we know,
  the first clean calculation of a p-value corresponding to a non-extreme value of a test statistic
  was that given by Daniel Bernoulli in his 1735 paper \cite{Bernoulli:1735}.
  In 1732, the Academy of Sciences of Paris had set a prize for the following problem:
  ``What is the physical cause of the inclination of the planes of the planetary orbits
  in relation to the plane of revolution of the Sun about its axis;
  and what is the reason for the inclinations of these orbits to differ among themselves''
  (translation by Hald \cite{Hald:1998}).
  None of the memoirs submitted was deemed to be worthy of the prize,
  and in 1734 the Academy proposed the same subject again with a double prize,
  which was awarded to John Bernoulli and his son Daniel.
  The key problem to explain was the small inclinations of the planetary orbits of all six known planets
  (Mercury, Venus, Earth, Mars, Jupiter, and Saturn)
  to the plane of the Sun's equator.
  Before looking for a physical cause,
  Daniel Bernoulli set out to investigate whether chance alone would be a feasible explanation.
  In modern terminology, his null hypothesis was that the six inclinations were chosen randomly
  (from the uniform distribution on the interval $[0,90^{\circ}]$),
  and his test statistic was their maximum.
  He computes the actual value of the test statistic as $7^{\circ}30'$, achieved for the Earth
  (according to 1701 data that he used),
  and finds the p-value as
  \[
    (7^{\circ}30'/90^{\circ})^6
    =
    1/12^6
    =
    1/2,985,984
  \]
  (\cite{Bernoulli:1735}, pp.~98--99, \cite{Todhunter:1865}, p.~223, \cite{Hald:1998}, p.~69).
  The way he states this result, however, is reminiscent of the modern confusion between p-values and probabilities:
  ``if all the orbits [orbital planes] were placed randomly with respect to the Sun's equator,
  I would bet 2985983 against 1 that they would not be so close.''
  We can criticize, on physical grounds, his choice of the null hypothesis,
  but from the statistical point of view his calculation is sound
  (which cannot be said about two other, more Earth-centered, p-values calculated in the same paper).

  Another influential (albeit less clean) calculation of a p-value was that in Laplace's 1823 paper \cite{Laplace:1823}
  devoted to atmospheric tides,
  a more difficult object of study than ocean tides.
  It is described in detail in, e.g., \cite[Section~4]{Stigler:1975} and \cite[Chapter~4]{Stigler:1986book},
  and it is given by David \cite{David:2001} as the first (?) appearance of the notion of p-value
  (``(?)'' standing for ``to his knowledge'').
  In the first version of this technical report that paper is cited
    as the first known to us p-value corresponding to a non-extreme value of a test statistic.

  Karl Pearson in his famous 1900 paper \cite{Pearson:1900} about the $\chi^2$ test
  and Ronald A. Fisher in his 1925 textbook \cite{Fisher:1925book}
  initiated the large-scale use of p-values.
  Pearson used $P$ as his notation for p-values,
  and on three occasions referred to it as ``the value of $P$''
  (crucially, the caption of his table of p-values for $\chi^2$
  given at the end of the paper is ``Values of $P$ for\ldots'').
  Fisher's textbook used ``value of $P$'' and simply ``$P$'' interchangeably.
  According to David \cite{David:1998},
  the ``value of $P$'' first (?) morphed into ``$P$ value'' in \cite[Remark on p.~30]{Deming:1943}
  (which is a great progress since his previous paper \cite{David:1995}, according to which this happened in a 1960 book).
  Deming's use of the form ``$P$ values''
  is consistent with his use of the notation $P(f)$ to refer to the p-values produced by a test statistic $f$
  (such as Karl Pearson's $\chi$ or Fisher's $z$).

  Randomized p-values might have been first introduced in print explicitly and in a fairly general form
  (for integer-valued test statistics) by Stevens in 1950 \cite[Section~4]{Stevens:1950}.
  Shortly before that, Anscombe in his discussion paper \cite{Anscombe:1948}
  had introduced randomized p-values in the special cases of Fisher's exact test
  (Section~5.07)
  and confidence bounds for the parameter of the Bernoulli model
  (Section~5.17, slightly less explicitly).
  But even at the time, this was not a novel idea:
  e.g., Egon Pearson in his 1950 paper \cite{Pearson:1950} defending randomized p-values says:
  ``The possibility of this conversion
  has been recognized by statisticians for a number of years''
  (Section~1);
  here ``conversion'' is his term for complementing the test statistic
  by a separate ``random experiment.''
  At about the same time, but less explicitly, randomized p-values were used
  by Eudey \cite{Eudey:1949} and Tocher \cite{Tocher:1950}.
  Randomized p-values were later used in the well-known book \cite{Pratt/Gibbons:1981} (Section~1.5.5).

\begin{thebibliography}{10}

\bibitem{Anscombe:1948}
Francis~J. Anscombe.
\newblock The validity of comparative experiments (with discussion).
\newblock {\em Journal of the Royal Statistical Society A}, 111:181--211, 1948.

\bibitem{Arbuthnott:1710}
John Arbuthnott.
\newblock An argument for divine providence, taken from the constant regularity
  observ'd in the births of both sexes.
\newblock {\em Philosophical Transactions of the Royal Society}, 27:186--190,
  1710.

\bibitem{Bernoulli:1735}
Daniel Bernoulli.
\newblock Recherches physiques et astronomiques sur le probl\`eme propos\'e
  pour la seconde fois par l'{A}cad\'emie {R}oyale des {S}ciences de {P}aris:
  {Q}uelle est la cause physique de l'inclinaison des plans des orbites des
  planetes par rapport au plan de l'\'equateur de la r\'evolution du soleil
  autour de son axe; {E}t d'o\`u vient que les inclinaisons de ces orbites sont
  diff\'erentes entre elles.
\newblock {\em Recueil des pi\`eces qui ont remport\'e les prix de l'Acad\'emie
  Royale des Sciences}, 3:95--122, 1735.
\newblock The original Latin text occupies pages 125--144 of this volume;
  according to the author's preface, his French translation on pages 95--122
  contains small additions and clarifications.

\bibitem{Birkhoff}
Garrett Birkhoff.
\newblock {\em Lattice Theory}.
\newblock American Mathematical Society, Providence, RI, third edition, 1967.
\newblock First edition: 1940. Second edition: 1948.

\bibitem{Cater}
Frank~S. Cater.
\newblock On order topologies and the real line.
\newblock {\em Real Analysis Exchange}, 25:771--780, 1999.

\bibitem{Coudin:2007}
Elise Coudin.
\newblock {\em Inf\'erence exacte et non param\'etrique dans les mod\`eles de
  r\'egression et les mod\`eles structurels en pr\'esence
  d'h\'et\'erosc\'edasticit\'e de forme arbitraire.}
\newblock PhD thesis, University of Montreal, 2007.

\bibitem{David:1995}
Herbert~A. David.
\newblock First (?) occurrence of common terms in statistics and probability.
\newblock {\em American Statistician}, 49:121--133, 1995.

\bibitem{David:1998}
Herbert~A. David.
\newblock First (?) occurrence of common terms in statistics and
  probability---a second list, with corrections.
\newblock {\em American Statistician}, 52:36--40, 1998.

\bibitem{David:2001}
Herbert~A. David.
\newblock First (?) occurrence of common terms in statistics and probability.
\newblock In Herbert~A. David and Anthony W.~F. Edwards, editors, {\em
  Annotated Readings in the History of Statistics}, pages 209--246 (Appendix
  B). Springer, New York, 2001.
\newblock This list subsumes the 1995 and 1998 lists
  \cite{David:1995,David:1998}.

\bibitem{Deming:1943}
W.~Edwards Deming.
\newblock {\em Statistical Adjustment of Data}.
\newblock Wiley, New York, 1943.

\bibitem{Deuchler:1914}
Gustav Deuchler.
\newblock \"uber die methoden der korrelationsrechnung in der p\"adagogik und
  psychologie.
\newblock {\em Zeitschrift f\"ur P\"adagogische Psychologie und Experimentelle
  P\"adagogik}, 15:114--131, 145--159, and 229--242, 1914.

\bibitem{Dickhaus}
Thorsten Dickhaus, Klaus Strassburger, Daniel Schunk, Carlos Morcillo-Suarez,
  Thomas Illig, and Arcadi Navarro.
\newblock How to analyze many contingency tables simultaneously in genetic
  association studies.
\newblock {\em Statistical Applications in Genetics and Molecular Biology},
  11(4):Article 12, 2012.

\bibitem{Eudey:1949}
Mark~W. Eudey.
\newblock {\em On the treatment of discontinuous random variables}.
\newblock PhD thesis, Statistical Laboratory, University of California,
  Berkeley, CA, 1949.
\newblock PhD thesis. Supervised by Jerzy Neyman.

\bibitem{Fisher:1925book}
Ronald~A. Fisher.
\newblock {\em Statistical Methods for Research Workers}.
\newblock Hafner, New York, 1925.

\bibitem{Gosset:1908}
William~S. {Gosset [Student]}.
\newblock The probable error of a mean.
\newblock {\em Biometrika}, 6:1--25, 1908.

\bibitem{G29}
Yuri Gurevich and Saharon Shelah.
\newblock Modest theory of short chains. {II}.
\newblock {\em Journal of Symbolic Logic}, 44:491--502, 1979.

\bibitem{Hald:1998}
Anders Hald.
\newblock {\em A History of Mathematical Statistics from 1750 to 1930}.
\newblock Wiley, New York, 1998.

\bibitem{Halmos:1960}
Paul~R. Halmos.
\newblock {\em Naive Set Theory}.
\newblock Van Nostrand, New York, 1960.

\bibitem{Lancaster:1961}
H.~Oliver Lancaster.
\newblock Significance tests in discrete distributions.
\newblock {\em Journal of the American Statistical Association}, 56:223--234,
  1961.

\bibitem{Laplace:1823}
Pierre-Simon Laplace.
\newblock De l'action de la lune sur l'atmosphere.
\newblock {\em Annales de Chimie et de Physique}, 24:280--294, 1823.

\bibitem{Lehmann:2006}
Erich~L. Lehmann.
\newblock {\em Nonparametrics: Statistical Methods Based on Ranks}.
\newblock Springer, New York, 2005.

\bibitem{Pearson:1950}
Egon~S. Pearson.
\newblock On questions raised by the combination of tests based on
  discontinuous distributions.
\newblock {\em Biometrika}, 37:383--398, 1950.

\bibitem{Pearson:1900}
Karl Pearson.
\newblock On the criterion that a given system of deviations from the probable
  in the case of correlated system of variables is such that it can be
  reasonably supposed to have arisen from random sampling.
\newblock {\em Philosophical Magazine}, 50:157--175, 1900.

\bibitem{Pratt/Gibbons:1981}
John~W. Pratt and Jean~D. Gibbons.
\newblock {\em Concepts of Nonparametric Theory}.
\newblock Springer, New York, 1981.

\bibitem{Rosenstein}
Joseph~G. Rosenstein.
\newblock {\em Linear Orderings}.
\newblock Academic Press, New York, 1982.

\bibitem{Stevens:1950}
Wilfred~L. Stevens.
\newblock Fiducial limits of the parameter of a discontinuous distribution.
\newblock {\em Biometrika}, 37:117--129, 1950.

\bibitem{Stigler:1975}
Stephen~M. Stigler.
\newblock Studies in the history of probability and statistics xxv. napoleonic
  statistics: the work of laplace.
\newblock {\em Biometrika}, 62:503--517, 1975.

\bibitem{Stigler:1986book}
Stephen~M. Stigler.
\newblock {\em The History of Statistics: The Measurement of Uncertainty before
  1900}.
\newblock Belknap Press of Harvard University Press, Cambridge, MA, 1986.

\bibitem{Tocher:1950}
Keith~D. Tocher.
\newblock Extension of the {N}eyman--{P}earson theory of tests to discontinuous
  variates.
\newblock {\em Biometrika}, 37:130--144, 1950.

\bibitem{Todhunter:1865}
Isaac Todhunter.
\newblock {\em A History of the Mathematical Theory of Probability from the
  Time of Pascal to that of Laplace}.
\newblock Macmillan, London, 1865.

\bibitem{Vovk/etal:2005}
Vladimir Vovk, Alex Gammerman, and Glenn Shafer.
\newblock {\em Algorithmic Learning in a Random World}.
\newblock Springer, New York, 2005.

\bibitem{Wilcoxon:1945}
Frank Wilcoxon.
\newblock Individual comparisons by ranking methods.
\newblock {\em Biometrics Bulletin}, 1:80--83, 1945.

\end{thebibliography}
\end{document}